\documentclass[11pt]{article}
\usepackage[]{amsmath,amssymb,amsfonts,latexsym,amsthm,enumerate,fullpage}
\newcommand{\remove}[1]{}

\def\F{\mathbb{F}}
\def\1{\mathbb{1}}
\def\R{\mathbb{R}}

\def\N{\mathbb{N}}

\newtheorem{remark}{Remark}
\newtheorem{thm}{Theorem}
\newtheorem{definition}[thm]{Definition}
\newtheorem{cor}{Corollary}

\newtheorem{question}{Question}

\newtheorem{proposition}[thm]{Proposition}

\begin{document}
\title{High Dimensional Expanders and Property Testing}
%\author{Tali Kaufman, Alexander Lubotzky}
\author{Tali Kaufman \thanks {Bar-Ilan University, ISRAEL. Email: \texttt{kaufmant@mit.edu}.
Research supported in part by the Alon Fellowship, and by an EU IRG grant.}
 \and Alexander Lubotzky \thanks {Hebrew University, ISRAEL. Email: \texttt{alexlub@math.huji.ac.il}.
Research supported in part by the ERC and by the Israel Science Foundation.}}
\maketitle
\begin{abstract}
We show that the high dimensional expansion property as defined by Gromov, Linial and Meshulam, for simplicial complexes is a form of testability.
Namely, a simplicial complex is a high dimensional expander iff a suitable property is testable.
Using this connection, we derive several testability results.
%We show that in the 2-dimensional Ramanujan complexes constructed by Lubotzky, Samuels and Vishne, every 1-co-cycle that belongs to a non-trivial co-homology class has a linear support.
\end{abstract}

\section{Introduction}
\subsection{High dimensional expanders}
Expander graphs have been playing an important role in computer science in the last few decades (see~\cite{HLW}) and more recently also in pure mathematics (see~\cite{Lub1}). In recent years a high dimensional theory of expanders is starting to emerge (see~\cite{Lub2} and the references therein). It is not even clear what is the "right" definition of expanders for simplicial complexes of dimension greater equal $2$.
But, two essentially equivalent definitions were given in two seminal works: One by Linial and Meshulam~\cite{LM} (see also~\cite{MW}) whose motivation was to study the (homological) connectivity of random complexes, as a first step toward developing a higher dimensional version of the Erdos-R\`{e}yni theory of random graphs. The second is by Gromov (see~\cite{Gromov1}) whose motivation was the study of fibers and overlapping properties of maps between complexes and manifolds. These two very different motivations led to a very similar definition, which we now give. We give a version that combines the two and which is most convenient for our needs. The homological/cohomological notions will be defined and explained in details in Section~\ref{sec:defs}.

%This theory grew out of three different motivations: On one hand a geometric/topological motivations led by Gromov (see~\cite{Gromov1, Gromov2, Wagner} and the references therein) and on the other hand a direction started by Linial and Meshulam~\cite{LM} of studying "random" simplicial complex as a high dimensional analogue of Erdos-Turan theory of random graphs. The third direction is the so-called "Ramanujan Complexes" which is an extension of the notion of Ramanujan graphs to higher dimensional simplicial complexes (see~\cite{LSV} and the references therein).
%
%The first two directions led to a quite similar definition of expansion of simplicial complexes which we now recall (The homological/cohomological notions will be defined and explained in details in Section~\ref{sec:defs}).

\begin{definition} \label{def-co-boundary-exp} Let $X$ be a finite simplicial complex of dimension $d$ and let $i \in \{0, \cdots, d-1 \}$. The $i$-th (coboundary) expansion constant $\epsilon_i$ of $X$ is defined as follows:
$$ \epsilon_i = \mbox{min} \frac{||\delta f ||}{||[f]|| }$$
where the minimum is taken over all $f \in C^{i}(X,\F_2) \backslash B^{i}(X,\F_2)$, where $C^{i}(X,\F_2)$ is the $\F_2$-vector space of all $i$-cochains of $X$, $B^{i}(X,\F_2)$ is the subspace of the coboundaries and $\delta f$ is the coboundary of $f$. The norm $||g||$ of $g \in C^{i}(X,\F_2)$ is the proportion of the $i$-cells on which $g$ does not vanish. The symbol $[f]$ stands for the coset of $f$ modulo $B^{i}(X,\F_2)$ and $||[f]|| = \mbox{min}\{||g|| \big{|} g \in [f]\}$.
\end{definition}

While this definition looks mysterious in first sight, one can check that for $i=d-1=0$, i.e., for graphs, it gives the standard (normalized) edge expansion ("The Cheeger constant") of graphs. Namely, for a graph $X=(V,E):$
$$\epsilon_0(X) = \frac{|V|}{|E|}\mbox{min}_{A \neq \emptyset,V}\frac{|E(A,\bar{A})|}{\mbox{min}\{|A|,\bar{|A|}\}}.$$

See Section~\ref{sec:defs} below for details.

A result proved independently by Meshulam and Wallach~\cite{MW} and by Gromov~\cite{Gromov1} is the following.

\begin{thm}\label{thm:thm2}
Let $X_n^{(d)}$ be the complete $d$-dimensional simplicial complex on $n$ vertices, i.e., the complex of all subsets of $[n] = \{1, \cdots, n\}$ of size at most $d+1$. Then, for every $i=0,\cdots, d-1$, $\epsilon_i \geq 1$.
\end{thm}

One of the difficult questions about simplicial complexes is to evaluate their expansion constants (see~\cite{DK} for some results in this direction). Theorem~\ref{thm:thm2} says informally that the complete $d$-dimensional simplicial complex is "an expander". It can be compared with the trivial result that this is the case for $d=1$ i.e., the complete graph is an expander. It is less trivial to show that there are bounded degree expander graphs, though by now various methods are known: random, Kazhdan property $T$, Ramanujan conjecture, the zig-zag product etc. An outstanding open problem is to show that higher dimensional bounded degree expanders exist. For some results in this direction see~\cite{Gromov1,KKL,LubM}.

The goal of this paper is to point out that this notion of high dimensional expansion is also of value and interest to theoretical computer science. %We will illustrate this by showing a connection to "property testing" as well as to "locally testable codes".
We show below that the above expansion is intimately related to the area of "Property Testing".

\subsection{Property Testing}
Let us recall first what it means for a property to be testable.

\begin{definition}\label{def-testability} (($q$,$\epsilon$)-testability) Let $A$ be a finite set, $W_n$ a subset of $A^n$ (the $n$-tuples of elements in $A$) and $P_n$ a subset of $W_n$. We say that the membership of $\alpha \in P_n$ (given $\alpha \in W_n$) is {\em testable}, or that $P_n$ is ($q$,$\epsilon$)-testable, if there exist $0 < \epsilon \in \R$, $q \in \N$ and a randomized algorithm, called a tester, which queries only $q$ (independent of $n$) coordinates of $\alpha$ and answers "yes" if $\alpha \in P_n$, while it answers "no" with probability at least $\epsilon \cdot \overline{dist}(\alpha,P_n)$ where $\overline{dist}(\alpha,P_n)$ is the normalized Hamming distance between $\alpha$ and the set $P_n$.
\end{definition}

One of the early works in the area of Property Testing is the work of Blum, Luby and Rubinfeld~\cite{BLR} which dealt with linearity testing (see~\cite{RS} for low degree testing).

{\bf Linearity Testing:} Let $W$ be the space of all functions from $\{0,1\}^m$ to $\{0,1\}$. This space is of dimension $n = 2^m$ over $\F_2$, the field of two elements, and let $W_0$ be the subspace of all linear functions, so the dimension of $W_0$ is $m$.

\begin{thm}~\cite{BLR}\label{thm:thm3}
\begin{itemize}
\item A function $f$ in $W$ is linear iff $f(x+y) = f(x) + f(y)$ for every $x,y \in \{0,1\}^m$.
\item There exists a constant $\epsilon > 0 $ such that for every $f \in W$
$$ \mbox{Prob}(f(x+y) \neq f(x) + f(y) | x,y \in \{0,1\}^m) \geq \frac{\epsilon \cdot dist(f,W_0)}{2^m}$$
where $dist(f,W_0)$  is the number of $x \in \{0,1\}^m$ on which $f$ has to be changed to make it linear (i.e., the Hamming distance between $f$ and $W_0$).
\end{itemize}
\end{thm}

The first item in Theorem~\ref{thm:thm3} is trivial (as this is the definition of a linear map!) and it also follows from the second item. The second part of Theorem~\ref{thm:thm3} is not trivial. It implies, in particular, that if we want to ensure with probability at least $1-\delta$ that a given $f \in W$ is linear (for some $\delta > 0$ e.g. $\delta = 0.001$), it suffices to check the equation $f(x+y) = f(x) + f(y)$ for a constant number $k_0$ of random inputs $x,y \in \{0,1\}^m$.
The number $k_0$ depends on $\delta$ but is independent of $n$. One can take $k_0 = \mbox{min}\{k | \epsilon^k  < \delta \}$.
In~\cite{BLR} it was proven that $\epsilon$ in Theorem~\ref{thm:thm3} can be taken to be $\epsilon=\frac{2}{9}$, and better estimates have been given~\cite{KIWI,KLX}.

So, Theorem~\ref{thm:thm3} says (when $A=\{0,1\}$, $n=2^m$, $W_n=A^n$ and $P_n$ the $m$-dimensional linear functionals) that the property of linearity is testable.

\subsection{Testability and Expansion}
The main point of this paper is the observation that expansion and testability are intimately connected with each other. The formal way to express it is Theorem~\ref{thm:main}, whose proof is obtained by spelling out carefully the definitions, but it requires the notions to be defined in Section~\ref{sec:defs}. Let us instead illustrate it here by a baby example.

\paragraph{The constant function property.} Let $\Gamma$ be a connected graph on $n$ vertices $[n] = \{1,\cdots,n\}$ and $f:[n] \rightarrow \{0,1\}$ a function on the vertices of $\Gamma$. Is $f$ a constant function?
Let us apply the following $\Gamma$-algorithm: choose a random edge of $\Gamma$ and check whether $f$ agrees on the two end points of the edge. Answer yes if it agrees and no otherwise.

\begin{proposition}\label{prop-constant-function-is-testable} The algorithm is a ($2$,$\epsilon$)-tester for the "constant function" property iff $\Gamma$ is an $\epsilon$-expander graph.
%More precisely:
%\begin{itemize}
%\item If $\Gamma_n$ is a family of graphs on $n$ vertices $n \rightarrow \infty$, which are all $\epsilon$-expanders, then there exists $\epsilon'=\epsilon'(\epsilon)$ for which the above algorithm is $\epsilon'$-tester
%\item If the algorithm is an $\epsilon$-tester for some $\epsilon >0$, for a family $\Gamma_n$ of graphs, then there exists an $\epsilon'>0$, such that such that these graphs are all $\epsilon'$-expanders.
%\end{itemize}
\end{proposition}

Before providing the proof of the proposition we recall that a graph $\Gamma=(V,E)$ is called an $\epsilon$-expander if for every subset $S$ of $V$
$$ \frac{|E(S,\bar{S})|}{|E|} \geq \epsilon \cdot \frac{min(S,\bar{S})}{|V|}.$$
Where $E(S,\bar{S})$ is the set of edges from $S$ to its complement $\bar{S}$. We are using this definition to take into account also graphs of unbounded degrees. For $k$-regular graphs ($k$ fixed) this is equivalent to the usual definition (up to a change of $\epsilon$).

\begin{proof} Note that the constant function property $P$ contains only two elements: the "all $0$" function and the "all $1$" function. Any function is a characteristic function of some subset $S$, $f=\chi_S$ and $\overline{dist}(f,P)=\frac{min(S,\bar{S})}{|V|}$. Now given such $f$ (and hence $S$) the proportion of edges that cause the tester to reject is exactly $\frac{|E(S,\bar{S})|}{|E|}$. So, the result follows immediately from the definitions.
\end{proof}

Let us end this introduction by recalling that property testing is closely related to locally testable codes (LTCs). The fact that LTCs are related to expander graphs was shown in~\cite{DinurK}. The results proven here show that if $\epsilon_i$, the $i$-coboundary expansion of $X$, is positive then $B^i(X,\F_2)$ is a locally testable code inside $C^i(X,\F_2)$. Unfortunately, as a code, $B^i(X,\F_2)$ has a poor distance, as it contains the image of every $(i-1)$-cell $\triangle$, which are vectors whose support is equal to the $deg(\triangle)=\#\{i\mbox{-cells containing } \triangle\}$. In most cases of interest this is relatively small or even bounded.
%(e.g. the Tensor power code parameters are $[{m \choose 2}, m-1,m-1]$). One may hope to improve the distance by replacing $B^i(X,\F_2)$ with $B^i(X,\F_2) \cap Z_i(X,\F_2)$. By analogy to real Hodge theory we call such codes "Harmonic codes". For example, if $X=(V,E)$ is a graph this gives the space of all even cuts of $X$, i.e., the space of all subsets $E' \subseteq E$ which meet every vertex an even number of times. The Harmonic code seems as an interesting object to study and we hope to come back to it in the future.

\section{Homology and Cohomology}\label{sec:defs}

\subsection{Definitions and basic facts}

In this section we will introduce the homological language needed in this paper. We will use only (co)homology with coefficients in the field $\F_2=\{0,1\}$ of two elements, which makes life easier than the general case as we can ignore orientation.

Let $X$ be a finite simplicial complex, i.e., $X$ is a non-empty collection of subsets of a finite set $V$, called the set of vertices, satisfying $F \in X$ and $G \subseteq F$ implies $G \in X$. In particular, $\emptyset \in X$. For a subset $F \in X$, denote $\mbox{dim}F=|F|-1$. If $\mbox{dim}F = i$ then $F$ is called an {\em $i$-face} (or a face of dimension $i$ or an $i$-cell).

The set of all $i$-faces is denoted $X(i)$, so $X(-1) = \{\emptyset\}$. We say that $X$ is of dimension $d$ if the face of largest size in $X$ is of dimension $d$ (i.e., of size $d+1$). A $1$-cell is called an edge, a $2$-cell a triangle etc.

Let us denote by $C_i = C_i(X,\F_2)$ the $\F_2$-vector space with basis $X(i)$ (or equivalently, the $\F_2$-vector space of subsets of $X(i)$), and $C^i = C^i(X,\F_2)$ the $\F_2$-vector space of functions from $X(i)$ to $\F_2$. It will be convenient sometimes to identify $C_i$ with $C^i$ in the obvious way. One can also think of $C^i$ as the dual of $C_i$.

The boundary map $\partial_i : C_i(X,\F_2) \rightarrow C_{i-1}(X,\F_2)$ is:
\begin{eqnarray}
\partial_i(F) = \sum_{G \subset F, |G|=|F|-1}G,
\end{eqnarray}
where $F \in X(i)$, and the coboundary map $\delta_i: C^i(X,\F_2) \rightarrow C^{i+1}(X,\F_2)$ is:
\begin{eqnarray}
\delta_i(f)(G)=\sum_{F \subset G, |F|=|G|-1}f(F),
\end{eqnarray}
where $f \in C^i$ and $G \in X(i+1)$.

Using the identification between $C^i$ and $C_i$ and defining the bilinear form $C^i \times C^i \rightarrow \F_2$ by:
%\begin{eqnarray}
$$\langle f,g  \rangle= \sum_{F \in X(i)} f(F) g(F).$$
%\end{eqnarray}
(all in the $\F_2$ arithmetic) we have
\begin{eqnarray}\label{eqnarray-inner-prod}
\mbox{ For $\alpha \in C^i = C_i$ and $\beta \in C_{i+1} = C^{i+1}$:} \langle \alpha, \partial \beta \rangle = \langle\delta \alpha , \beta\rangle.
\end{eqnarray}
Indeed to prove (\ref{eqnarray-inner-prod}), it is sufficient to do it for $\alpha \in X(i)$ and $\beta \in X(i+1)$ and in such a case both sides of  (\ref{eqnarray-inner-prod}) are $1$ if $\alpha \subseteq \beta$ and $0$ otherwise.

Well known and easily calculated equations are:
\begin{eqnarray}\label{eqnarray-double-partial}
\partial_i \circ \partial_{i+1}=0 \mbox{ and } \delta_{i+1} \circ \delta_{i} = 0
\end{eqnarray}
Or if we omit subscripts $\partial^2 =0$ and $\delta^2=0$.

Thus, if we denote:

$B_i = B_i(X,\F_2) = \mbox{Image}(\partial_{i+1})=$\mbox{ the space of $i$-boundaries}.

$Z_i = Z_i(X,\F_2) = \mbox{Ker}(\partial_{i})=$\mbox{ the space of $i$-cycles}.

$B^i = B^i(X,\F_2) = \mbox{Image}(\delta_{i-1})=$\mbox{ the space of $i$-coboundaries}.

$Z^i = Z^i(X,\F_2) = \mbox{Ker}(\delta_{i})=$\mbox{ the space of $i$-cocycles}.

We get from (\ref{eqnarray-double-partial})
\begin{eqnarray}
B_i \subseteq Z_i \subseteq C_i \mbox{ and } B^i \subseteq Z^i \subseteq C^i.
\end{eqnarray}

Define the quotient spaces $H_i(X,\F_2) =Z_i/B_i$ and $H^i(X,\F_2) =Z^i/B^i$, the $i$-homology and the $i$-cohomology groups of $X$ (with coefficients in $\F_2$).

A linear algebra exercise (as $\F_2$ is a field) shows that

$\mbox{dim}H_i(X,\F_2)=\mbox{dim}H^i(X,\F_2)$. Another easy exercise gives that when $X=(V,E)$ is a graph with a set of vertices $V$ and a set of edges $E$ then $\delta_{-1}:C^{-1}(X,\F_2) \rightarrow C^0(X,\F_2)$ sends the one-dimensional space $C^{-1}$ to $B^0 = \{0, \bf{1} \}$ where $\bf{1}$ is the "all" function giving $1$ to every vertex. Furthermore, $\delta_0: C^0 \rightarrow C^1$ sends every vertex to the "star" around it. It follows that a subset $D$ of $V$ is in $\mbox{Ker}\delta_0$ iff $D$ is a union of connected components of the graph. Hence, $\mbox{dim}H^0 = \mbox{dim}Z^0 - \mbox{dim}B^0 = -1 + \mbox{\# connected  components of $X$}$.
Hence:
\begin{eqnarray}
X \mbox{ is connected } \Leftrightarrow \mbox{dim}H^0(X,\F_2) = 0
\end{eqnarray}

Assuming $X$ is connected, an easy computation shows that $\mbox{dim}H^1(X,\F_2) = |E|-|V|-1$. So, $H^1(X,\F_2)=0$ iff $X$ is a tree.

Going now back to general $X$ we can deduce from (\ref{eqnarray-inner-prod}) and (\ref{eqnarray-double-partial}) that
\begin{eqnarray}\label{eq:dual-to-co-boundaries}
B_i^{\bot} = Z^i \mbox{ and } Z_i^{\bot} = B^i.
\end{eqnarray}

Indeed, if $\beta \in B_i$ and $\alpha \in Z^i$ then $\beta = \partial \gamma$ for some $\gamma \in C_{i+1}$ and so
$\langle\beta, \alpha\rangle = \langle\partial \gamma, \alpha\rangle = \langle\gamma, \delta \alpha\rangle = \langle\gamma, 0\rangle = 0$, hence $Z^i \subseteq B_i^{\bot}$.
On the other hand, if $\alpha \in C^i$ and for every $\gamma \in C^{i+1}$, $\langle\partial \gamma, \alpha\rangle=0$, then for every $\gamma \in C^{i+1}$, $\langle\gamma, \delta \alpha\rangle =0$. The bi-linear form is non-degenerate and hence $\delta \alpha =0$ i.e., $\alpha \in Z^i$. In a similar way we deduce also the second equality.

\subsection{$\F_2$-coboundary expansion}

We are now ready to define following~\cite{LM,Gromov1} the expansion of a simplicial complex. Let $X$, as before, be a finite simplicial complex of dimension $d$ and $\alpha \in C^i(X,\F_2)$ for some $i$, $0 \leq i \leq d$.

Denote:
$|\alpha| = \#\{F \in X(i) | \alpha(F) \neq 0 \}$ and $||\alpha|| = \frac{|\alpha|}{|X(i)|}$, i.e., the proportion of the number of $i$-cells on which $\alpha$ does not vanish. Of course, in our case, $\alpha(F) \neq 0$ means $\alpha(F) = 1$, but one may want to think also about this notion over other fields. In the case of $\F_2$, we can also think of $\alpha$ as simply a subset of $X(i)$, and $|\alpha|$ is its order.

For $\alpha \in C^i = C^i(X,\F_2)$ and $W$ a subspace of $C^i$ let:
$$\mbox{dist}(\alpha, W) = \mbox{min}\{|\alpha_0| \mbox{ } | \alpha_0 \in \alpha+W\}.$$

It is easy to see that if the minimum is obtained on $\bar{\alpha}$ in the coset $\alpha+W$ then $\alpha- \bar{\alpha}$ is the closest element to $\alpha$ in $W$ in the Hamming distance and the distance is indeed $|\bar{\alpha}|$. We should remark that this $\bar{\alpha}$ (and hence also $\alpha- \bar{\alpha}$) is not necessarily unique. This is a significant difference between "geometry over $\F_2$" versus "geometry over $\R$".

We also define the normalized distance:
$$\overline{dist}(\alpha,W) = ||\bar{\alpha} || =\frac{\mbox{dist}(\alpha, W)}{|X(i)|}.$$

Let us give now Definition~\ref{def-co-boundary-exp} in a slightly different form which will be more convenient for us.

\begin{definition}($\F_2$-coboundary expansion)\label{def-coboundary expansion}
For $i=0,\cdots,d-1$ denote:
$$\epsilon_i(X) = \mbox{min} \frac{||\delta_{i} f ||}{\overline{dist}(f,B^{i}(X,\F_2))},$$
where the minimum is taken over all the functions $f \in C^{i} \setminus B^{i}$.
\end{definition}

Definition~\ref{def-coboundary expansion} is equivalent to Definition~\ref{def-co-boundary-exp} as $dist(\alpha,W)$ is equal to the minimum Hamming weight among the elements of the coset $\alpha+W$.

\begin{remark} The $\F_2$-coboundary expansion is defined essentially as in~\cite{LM} and~\cite{Gromov1}, but is slightly different from both. In~\cite{LM}, Linial and Meshulam studied the quotients $\mbox{min}\frac{|\delta_i f|}{|[f]|}$, without giving it a name, with a goal to prove that $H^i$ vanishes. This is as our $\epsilon_i$ up to normalization factors. Gromov in~\cite{Gromov1} studied $$\mu_i = \mbox{max}_{0 \neq \beta \in B^{i+1}(X,\F_2)}  (\frac{1}{||\beta||} \cdot {\mbox{min}_{\alpha \in C^i(X,\F_2), \delta \alpha = \beta} ||\alpha||}).$$ It is easy to see that if $H^i=0$ then $\mu_i = \frac{1}{\epsilon_i}$, but it is possible that $\mu_i < \infty$ even if $H^i \neq 0$ and $\epsilon_i = 0$. The name "coboundary expansion" was coined in~\cite{DK}.
\end{remark}

An easy corollary of the definition is the following.

\begin{cor} $\epsilon_i(X) \neq 0 \mbox{ iff } H^{i}(X,\F_2)=0$.
\end{cor}

\begin{proof} If $H^{i}(X,\F_2) \neq 0$ then there exists $f \in  C^{i} \setminus B^{i}$ with $\delta_{i} f = 0$, hence, $\epsilon_i(X)=0$. On the other hand, if $H^{i}(X,\F_2) = 0$ then $\epsilon_i(X)$ is a minimum over a finite set of rational numbers each of which is non-zero.
\end{proof}

Let us spell out the definition for graphs. Here, $X=(V,E)$ and $d=1$. A function $f \in C^0(X,\F_2)$ is a characteristic function $1_A$ of some subset $A$ of $V$, $B^0(X,\F_2)$ is, as explained above, the $1$-dimensional space of $0$ and $1_V$.
Thus $\overline{dist}(1_A,B^0) = \frac{1}{|V|}\mbox{min}(|A|, |V|\setminus|A|)$. On the other hand, one can easily check that $\delta_0(1_A) = E(A,\bar{A})$, i.e., the characteristic function of the set of all edges from $A$ to its complement.
Hence,
$$\epsilon_0(X) = \mbox{min}_{\emptyset \neq A \varsubsetneq V } \frac{\frac{1}{|E|} |E(A,\bar{A})|}{\frac{1}{|V|} \mbox{min}(|A|, |\bar{A}|)} = \frac{|V|}{|E|}\mbox{min}_{A \neq \emptyset,V}\frac{|E(A,\bar{A})|}{\mbox{min}\{|A|,\bar{|A|}\}}.$$

This gives us (up to the normalized factor $\frac{|V|}{|E|}$) the standard "edge expansion" (known also as the Cheeger constant) which defines expander graphs. Note also that when $X$ is a $k$-regular graph, $k$ fixed, $\frac{|V|}{|E|} = \frac{2}{k}$ is a constant.

\subsection{The expansion of complete complexes}\label{subsec:exp-complete-complexes}

Theorem~\ref{thm:thm2} of the introduction was proved independently by Linial and Meshulam~\cite{LM} for $d=2$, by Meshulam and Wallach~\cite{MW} and Gromov~\cite{Gromov1}, independently, for general $d$. It states:

{\bf Theorem 2} {\em Let $X_n^{(d)}$ be the complete $d$-dimensional simplicial complex on $n$ vertices, i.e., the complex of all subsets of $[n] = \{1, \cdots, n\}$ of size at most $d+1$. Then, for every $i=0,\cdots, d-1$, $\epsilon_i (X) \geq \frac{n}{n-i-1} \geq 1$.}

To get the flavor of this result, let us bring here the proof for $d=2$. This is the case which we mostly use in Section~\ref{sec-exp-and-tst}, so it will make the proofs in the current paper self contained.

Computing $\epsilon_0(X)$ in this case is easy as it depends only of the $1$-skeleton of $X$ and it is (as shown before) the normalized edge expansion, i.e.,

$$ \epsilon_0(X) = \frac{|X(0)|}{|X(1)|}\mbox{min}_{A \neq \emptyset,X(0)}\frac{|E(A,\bar{A})|}{\mbox{min}\{|A|,\bar{|A|}\}}.$$

In our case the minimum is obtained for $|A| = \frac{n}{2}$, hence,

$$ \epsilon_0(X) \geq \frac{n}{{n \choose 2}} \frac{\frac{n}{2} \cdot \frac{n}{2}}{ \frac{n}{2}} = \frac{n}{n-1} \geq 1.$$

Evaluating $\epsilon_1(X)$ is more involved. Let $\alpha \in C^1(X,\F_2)$, so $\alpha$ is a set of edges (or a function on the edges). For an edge $e \in X(1)$ and a vertex $u \in X(0)$ with $u \notin e$, denote by $ue$ the triangle formed by $u$ and $e$. Given $\alpha$, define $\alpha_u \in C^0(X,\F_2)$, "the local view of $\alpha$ from $u$" by:
$$\alpha_u(v) = \alpha((u,v)) \mbox{ if } u \neq v, \mbox{ and } \alpha_u(v) = 0 \mbox{ otherwise }.$$
One can now checks that for every edge $e$:

\begin{eqnarray}\label{eqnarray-LM}
(\alpha-\delta_0 \alpha_u )(e) = (\delta_1 \alpha) (u e) \mbox{ if } u \notin e, (\alpha - \delta_0 \alpha_u )(e) =0 \mbox{ if } u \in e
\end{eqnarray}

It now follows that:
\begin{eqnarray}
3 |\delta_1 \alpha| & = & |\{ (u,T) \in X(0) \times X(2) | u \in T, T \in \delta_1 \alpha\}| \\
                      & = & |\{(u,e) \in  X(0) \times X(1) | e \in \alpha - \delta_0 \alpha _u\}| \\
                      & = &  \sum_{u \in X(0)}|\alpha - \delta_0 \alpha_u| \geq n \cdot \mbox{dist}(\alpha, B^1).
\end{eqnarray}

The first and the third equalities are just a careful spelling out of the notations, while the second follows from (\ref{eqnarray-LM}). The inequality follows from the fact that $\delta_0\alpha_u \in B^1$. We deduce that $\frac{|\delta_1 \alpha|}{\mbox{dist}(\alpha, B^1)} \geq \frac{n}{3}$.
When we normalize by $|X(2)| = {n \choose 3}$ and $|X(1)| = {n \choose 2}$, we get:
$$ \frac{||\delta_1 \alpha ||}{\overline{dist}(\alpha, B^1)} \geq \frac{n}{n-2} \geq 1.$$

%\section{Property Testing}\label{sec:tensor-power}
%In this section we will show how Theorem~\ref{thm:thm2} can be applied to prove "testability" of three (seemingly) different problems. In fact, we will use here only the special case $d=2$ of Theorem~\ref{thm:thm2}, a case for which a complete proof was given in Section~\ref{subsec:exp-complete-complexes}.
%
%Let us recall first what it means for a property to be testable.
%
%{\bf Definition~\ref{def-testability}.} (($q$,$\epsilon$)-testability) Let $A$ be a finite set, $W_n$ a subset of $A^n$ (the $n$-tuples of elements in $A$) and $P_n$ a subset of $W_n$. We say that the membership of $\alpha \in P_n$ (given $\alpha \in W_n$) is {\em testable}, or that $P_n$ is testable, if there exists $0 < \epsilon \in \R$, $q \in \N$ and a randomized algorithm, called a tester, which queries only $q$ (independent of $n$) coordinates of $\alpha$ and answers "yes" if $\alpha \in P_n$, while it answers "no" with probability at least $\epsilon\overline{dist}(\alpha,P_n)$ where $\overline{dist}(\alpha,P_n)$ is the normalized Hamming distance between $\alpha$ and the set $P_n$.

\section{Expansion and Testability}\label{sec-exp-and-tst}
We are ready now to give the main result, which shows that expansion is a form of testability.

\begin{definition} (Cocycle Tester) Let $X$ be a simplicial complex and $f \in C^i(X,\F_2)$. The {\em $i$-cocycle tester} for $f$ is the test that picks a random $(i+1)$-cell $F$ and evaluate $\delta(f)(F)$. The test accepts $f$ iff $\delta(f)(F)=0$.

Note that, the number of queries performed by the $i$-cocycle tester is $i+2$ since every ($i+1$)-cell has exactly $i+2$ sub-cells of dim $i$ (obtained by deleting one of its vertices).
\end{definition}

Our main result is that high-dimensional (coboundary) expansion is equivalent to the local testability of the space of coboundaries. Namely:

\begin{thm}~\label{thm:main}(Main)
Let $X$ be a $d$ dimensional simplicial complex and let $i \in \{0, \cdots, d-1\}$. Then its $i$-th (coboundary) expansion
is $\epsilon_i$ iff the subspace $B^{i}(X,\F_2) \subseteq C^{i}(X,\F_2)$ is $($i+2$,\epsilon_i$)-testable via the "$i$-cocycle tester".
\end{thm}

\begin{remark} Proposition~\ref{prop-constant-function-is-testable} is a special case of Theorem~\ref{thm:main} when $d=1$ and $i=0$.
\end{remark}

\begin{proof} Assume first that the $i$-cocycle tester is an $(i+2, \epsilon)$ tester for $B^i(X,\F_2)$ inside $C^i(X,\F_2)$. This means that the probability for a function $f \in C^i(X,\F_2)$ to fail the test is at least $\epsilon$ times its normalized Hamming distance from $B^i(X,\F_2)$. (In particular, it says that if $f$ passes the test with probability $1$, i.e., if $f \in Z^i(X,\F_2)$, then it is in $B^i(X,\F_2)$, hence $B^i(X,\F_2)= Z^i(X,\F_2)$, and so  $H^i(X,\F_2)=0$.)

This means that
$$\frac{\mbox{\#} \{F \in X{(i+1)} | \delta(f)(F) \neq 0\}}{|X{(i+1)}|} \geq \epsilon \cdot \overline{dist}(f, B^i(X,\F_2)).$$
which, by Definition~\ref{def-coboundary expansion}, means exactly that the $i$-coboundary expansion of $X$ is at least $\epsilon$.

Conversely, assume that the $i$-coboundary expansion of $X$ equals $\epsilon$, then for every function $f \in C^i(X,\F_2)\setminus B^i(X,\F_2)$,
$\frac{|| \delta f||}{\overline{dist}(f, B^i(X,\F_2))} \geq \epsilon$, i.e.,
$$\frac{\mbox{\#} \{F \in X{(i+1)} | \delta(f)(F) \neq 0\}}{|X{(i+1)}|} \geq \epsilon \cdot \overline{dist}(f, B^i(X,\F_2)).$$
(see Definition~\ref{def-coboundary expansion}). This exactly means that the probability of $f$ to fail the $i$-cocycle test is at least $\epsilon$ times its (normalized) distance from $B^i(X,\F_2)$. Thus, the $i$-cocycle test is an $(i+2,\epsilon)$-test for $B^i(X,\F_2)$ inside $C^i(X,\F_2)$.
\end{proof}

Theorem~\ref{thm:main} says that any expansion result for some simplicial complex gives automatically a testability result. Let us illustrate this by some applications of Theorem~\ref{thm:thm2}.

\subsection{Sum functions on graphs}
Let $\Gamma$ be a finite graph on a set of vertices $[m]=\{1,\cdots,m\}$, and set of edges $E \subseteq {[m] \choose 2}$ of size $n$. A function $f:E \rightarrow \{0,1\}$ will be called a {\em sum function} if there exists a function $g:[m] \rightarrow \{0,1\}$ such that for every edge $e=(i,j)\in E$, $f(e)=g(i)+g(j) \mbox{ (mod }2)$.

\begin{question} Is the sum-function property testable?
%\tnote{show that for the cycle graph the sum function property is not testable.}
\end{question}

This depends on the structure of the graph. Let us start with a negative result.
\begin{proposition}\label{prop-sum-function-is-not-testable}
Let $\{X_i = (V_i,E_i)\}_{i \in I}$ be a family of finite graphs whose girth is going to infinity. Then, the sum function property on functions $f:E_i \rightarrow \{0,1\}$ is not testable.
\end{proposition}
\begin{proof} In the notations above, the sum functions is exactly the space $B^1(X_i,\F_2)$. This is a linear space (or a code) and it is well known (see e.g.~\cite{BHR}) that such a code is locally testable only if it is an LDPC (Low Density Parity Check) code. I.e., its dual space is spanned by bounded weight constraints. However, in our case, $B^1(X_i,\F_2)^{\bot}=Z_1(X_i,\F_2)$ (see equation (\ref{eq:dual-to-co-boundaries})). But, $Z_1(X_i,\F_2)$, the space of cycles, has no bounded weight vectors as the girth of $X_i$ is unbounded.
\end{proof}
We now show that for the complete graphs the answer is positive.

\begin{proposition}\label{prop-sum-function-is-testable}
Let $\Gamma = K_m$ the complete graph on $m$ vertices. Then the sum-function property on $\Gamma$ is testable.
\end{proposition}

\begin{proof}
We embed $\Gamma = K_m$ as the $1$-skeleton of $X = K_m^{(2)}$ , the $2$-dimensional complete complex on $m$ vertices (i.e., the set of all subsets of $[m]$ of cardinality at most $3$). Note that the space of sum functions is exactly $B^1(X,\F_2)$ - the space of coboundaries, as can be easily seen by spelling out the definitions. By Theorem~\ref{thm:thm2}, $\epsilon_1(K_m^{(2)}) \geq 1$. This means by Theorem~\ref{thm:main}, that the $1$-cocycle tester is a $(3,1)$-tester for $B^1(X,\F_2)$ inside $C^1(X,\F_2)$. The meaning of the $1$-cocycle tester is: choose a random tuple $\{r,j,k\} \in {[m] \choose 3 }$, accept a function $f$ on the edges of $\Gamma = K_m$ (i.e. on $X(1)$) iff $f(rj)+f(jk)+f(kr) = 0$. Indeed, it performs $3$ queries and the fact that it is a $(3,1)$-tester for $B^1(X,\F_2)$ exactly means that the probability of the test to reject $f$ is at least its (normalized) distance from $B^1(X,\F_2)$, i.e., from the sum-functions. This proves the proposition.

\end{proof}

\begin{remark} Note that the proposition says in particular that a function on the edges of $\Gamma= K_m$ is a sum function iff it vanishes on triangles. This could be proved directly, but a conceptual way to see this is the following:
$X= K_m^{(2)}$ is a triangulation of a bouquet of $b$ two-dimensional spheres $Y=\bigvee_{i=1}^{b}S^2$, $S^2 =\{(x,y,z) \in \R^3 | x^2+y^2+z^2 = 1\}$ . It has, therefore, the same cohomology/homology as $S^2$. %$S^2$ is simply connected,
Now, $H^1(S^2,\F) =0$ for every field $\F$, and, in particular, $H^1(S^2,\F_2) =0$. Hence also $H^1(Y,\F_2)=\bigoplus_{i=1}^b H^1(S^2,\F_2) =0$, and so $H^1(X,\F_2)=0$. This exactly means that every $1$-cocycle of $X$ is a coboundary. In particular, a function on the edges of $\Gamma= K_m$ is a sum function iff it vanishes on all triangles.
\end{remark}

%\tnote{add elementary direct proof?}
%Part $(b)$ of Theorem~\ref{thm:thm4} is exactly a reformulation of the result in Theorem~\ref{thm:thm2} (for $d=2$), that the expansion constant $\epsilon_2$ of $X_m^{(2)}$ is at least $1$.
%
%Now Theorem~\ref{thm:thm4} implies Proposition~\ref{prop-sum-function-is-testable}:
%
%\begin{proof}(of Proposition~\ref{prop-sum-function-is-testable})
%Indeed the algorithm will be the "triangle checking algorithm": Choose a random triangle $\{i,j,k\}$ and check if $f(i,j)+f(j,k)+f(k,i)=0$. If this holds answer "yes" otherwise "no". Theorem~\ref{thm:thm4} ensures that Proposition~\ref{prop-sum-function-is-testable} holds with $\epsilon=1$.
%\end{proof}

The proof of Proposition~\ref{prop-sum-function-is-testable} shows.

\begin{proposition} Let $X$ be a two dimensional simplicial complex and $\Gamma$ its $1$-skeleton. Then the "triangle checking algorithm" of $X$ is a $($3$,\epsilon$)-tester for the sum-function problem of $\Gamma$ if and only if the first expansion constant of $X$ is at least $\epsilon$.
\end{proposition}

%Before moving to another application of Theorem~\ref{thm:main}, let us show that for some graphs the sum function property is not testable.

%This is of course more general and if we allow ourselves to use the language of co-homology we obtain our main theorem.
%
%
%{\bf Theorem~\ref{thm:main}.}(Main Theorem)
%Let $X$ be a $d$ dimensional simplicial complex and let $i \in \{1, \cdots, d\}$. Then its $i$-th (co-boundary) expansion constant
%is $\epsilon_i$ iff the subspace $B^{i-1}(X,\F_2)) \subset C^{i-1}(X,\F_2))$ is $($i+1$,\epsilon_i$)-testable via the "co-cycle tester" (i.e., for $f \in C^{i-1}(X,\F_2)$, choose random $i$-dimensional cell $F$ of $X$ and evaluate $\delta(f)(F)$).

%\begin{proposition} Let $X$ be a $d$ dimensional simplicial complex then the subspace $B^{d-1}(X,\F_2)) \subset C^{d-1}(X,\F_2))$ has an $\epsilon$-tester via the "co-cycle algorithm" (i.e., for $f \in C^{d-1}(X,\F_2)$, choose random $d$-dimensional cell $F$ of $X$ and evaluate $\delta(f(F))$.) iff $\epsilon \geq \epsilon_d(X)$.
%\end{proposition}

\subsection{Tensor power testing}\label{sub-sec-tensor-power-testing}
Proposition~\ref{prop-sum-function-is-testable} gives also a testability result of a different form (We are grateful to Irit Dinur for calling our attention to this fact).

%The property testable by Proposition~\ref{prop-sum-function-is-testable} may look a bit artificial (and indeed it was invented for applying Theorem~\ref{thm:thm4}). But as pointed out to us by Irit Dinur, if one just switch notations from $\{0,1\}$ to $\{1,-1\}$ we get the following.

Let $A=\{1,-1\}$ and $A^{\frac{m^2-m}{2}}$ is the set of all $m \times m$ symmetric matrices with $1$ along the diagonal and $+1/-1$ outside the diagonal. Let $P_m$ be the subset of all matrices $M$ obtained as a tensor power, i.e., there exists a vector $\alpha$ of length $m$ with $+1/-1$ entries, such that $M_{i,j} =\alpha_i \cdot \alpha_j$.

\begin{proposition} $P_m$ is testable within $A^{\frac{m^2-m}{2}}$ via the algorithm: choose three different indices $\{i,j,k\}$, answer "yes" if
$M_{i,j}M_{j,k}M_{k,i} =1$ and "no" otherwise.
\end{proposition}

\begin{proof} This is just a reformulation of Proposition~\ref{prop-sum-function-is-testable} when one switch notations from $\{0,1\}$ to $\{1,-1\}$. With these notations the tensor powers are exactly the sum-functions and $A^{\frac{m^2-m}{2}}$ is the space of all symmetric functions on $[m] \times [m]$ with zeros along the diagonal.
\end{proof}

%
%It is natural to ask whether one can test if an $m \times m$ matrix $M$ of  $+1/-1$ entries is a tensor product of two $m$-vectors $\alpha,\beta$. i.e., whether there exist $\alpha,\beta$ such that  $M_{i,j} =\alpha_i \cdot \beta_j$ for every $1 \leq i,j \leq m$. This property is also testable but showing this require extension of the notions of expansions from simplicial complexes to cubic complexes. This will be done in the longer version of this paper.
%\tnote{Should we refer we tensor testing?}

\subsection{Seidel Switching}

There is another interpretation for Proposition~\ref{prop-sum-function-is-testable}.
Given a graph $\Gamma$ with a set of vertices $[n]=\{1,\cdots,n\}$ and let $v$ be a vertex of $\Gamma$. One defines the {\em Seidel switching} of $\Gamma$ along $v$ to be the graph obtained from $\Gamma$ by deleting all the edges of $\Gamma$ incident to $v$ and connecting $v$ to all vertices of $[n]$ which were not its neighbors in $\Gamma$. We say that the graph $\Gamma'$ on $[n]$ is {\em Seidel equivalent} to $\Gamma$ if it can be obtained from $\Gamma$ by a sequence of such Seidel switchings. This is indeed an equivalence relation. The concept of Seidel switching appears in various areas of combinatorics and computer science. See for example~\cite{KNZ} and the references therein.

We can now prove the following

\begin{proposition} The property of Seidel equivalence of a pair of graphs is testable. I.e., given two graphs $\Gamma$,$\Gamma'$ on $[n]$, there is a tester for the property that the pair of graphs are Seidel equivalent.
\end{proposition}

\begin{proof} A graph $\Gamma$ (resp $\Gamma'$) on $[n]$ can be considered as a symmetric function from the set of $2$-elements subsets of $[n]$ to $\{0,1\}$, i.e., as a function on the edges of the complete graph $K_n$, so it is a $1$-cochain of $K_n$. Now, one can see that if $\alpha$ is the cochain associated to $\Gamma$, then $\alpha+\delta_0(\chi_v)$ is the cochain associated with the Seidel switching of $\Gamma$ along $v$ (where $\chi_v$ is the characteristic function of $\{v\}$ and $\delta_0$ the coboundary map).

As $B^1(K_n, \F_2)$ is generated by $\delta_0(\chi_v)$, $v \in [n]$, it follows that $\Gamma'$ is Seidel equivalent to $\Gamma$ iff $\alpha'$ , the cochain associated with $\Gamma'$, is in the same coset modulo $B^1(K_n, \F_2)$ as $\alpha$, i.e., iff $\alpha-\alpha' \in B^1(K_n, \F_2)$.
Theorem~\ref{thm:thm2} implies, as explained in the proof of Theorem~\ref{thm:main} and Proposition~\ref{prop-sum-function-is-testable}, that the question of whether a cochain of $K_{n}$ is a coboundary, is testable. Applying this to $\alpha - \alpha'$ we deduce that "Seidel equivalence" is also testable. In fact the tester acts as follows. "Pick a random triangle and check whether $\delta_1(\alpha)$ and $\delta_1(\alpha')$ agree on it".
\end{proof}

Let us warn the reader that the last proposition says that the Seidel equivalence of graphs $\Gamma$ and $\Gamma'$ is testable when we consider them as labeled on the vertices $[n]$, and the Seidel switching preserves the label of the vertices. We do not expect the abstract Seidel switching to be testable since the decision version of that problem is equivalent to the problem of graph isomorphism~\cite{KNZ}.

\end{document}